\documentclass[11pt,notitlepage]{article}
\usepackage{pset}
\title{Pseudorandomness for Read-Once, Constant-Depth Circuits}
\author{
\begin{tabular}{*{3}{>{\centering}p{.31\textwidth}}}
\large Sitan Chen & \large Thomas Steinke\thanks{Supported by NSF grant CCF-1420938} & Salil Vadhan\thanks{Supported by NSF grant CCF-1420938 and a Simons Investigator grant} \tabularnewline
{\footnotesize \url{sitanchen@college.harvard.edu}} & {\footnotesize\url{tsteinke@seas.harvard.edu}} & {\footnotesize\url{salil@seas.harvard.edu}}
\end{tabular}}

\makeatletter
\@addtoreset{case}{section}
\makeatother

\begin{document}

\maketitle

\begin{abstract}
For Boolean functions computed by read-once, depth-$D$ circuits with unbounded fan-in over the de Morgan basis, we present an explicit pseudorandom generator with seed length $\tilde{O}(\log^{D+1} n)$. The previous best seed length known for this model was $\tilde{O}(\log^{D+4} n)$, obtained by Trevisan and Xue (\emph{CCC `13}) for all of $\AC^0$ (not just read-once). Our work makes use of Fourier analytic techniques for pseudorandomness introduced by Reingold, Steinke, and Vadhan (\emph{RANDOM `13}) to show that the generator of Gopalan et al. (\emph{FOCS `12}) fools read-once $\AC^0$. To this end, we prove a new Fourier growth bound for read-once circuits, namely that for every $F: \{0,1\}^n\to\{0,1\}$ computed by a read-once, depth-$D$ circuit, \begin{equation*}\sum_{s\subseteq[n], |s|=k}|\hat{F}[s]|\le O(\log^{D-1}n)^k,\end{equation*} where $\hat{F}$ denotes the Fourier transform of $F$ over $\mathds{Z}^n_2$.
\end{abstract}

\section{Introduction}
\label{sec:intro}

\subsection{Pseudorandomness for Constant-Depth Circuits}

A central question in pseudorandomness is whether the class of all decision problems solvable in randomized polynomial time can also be solved in deterministic polynomial time ($\P \stackrel{?}{=} \BPP$). To resolve this in the affirmative, it suffices to show that there exist logarithmic-seed-length pseudorandom generators that fool polynomial-size circuits, where a generator $G:\{0,1\}^m\to\{0,1\}^n$ is said to \textbf{$\varepsilon$-fool} a function $F: \{0,1\}^n\to\{0,1\}$ if $$|\E[F(U_n)] - \E[F(G(U_m))]|\leq\varepsilon.$$ Such generators were constructed by Impagliazzo and Wigderson \cite{iw} under the assumption that there are exponential time decision problems that require circuits of exponential size.

To obtain \emph{unconditional} results in pseudorandomness, however, it becomes necessary to restrict the class of ``distinguishers'' that a generator should fool. Ajtai and Wigderson \cite{aw} were the first to consider the problem of constucting generators specifically for $\AC^0$, i.e. \emph{constant-depth} circuits with unbounded fan-in over the de Morgan basis (AND, OR, and NOT gates), and in their pioneering work they achieved seed length $O(n^{\varepsilon})$ for any constant $\varepsilon>0$. 

Nisan \cite{nis} then improved this seed length to $\polylog(n)$ using hardness of parity for $\AC^0$. Subsequent works \cite{baz, brav, dett, raz} have used bounded independence or small-bias spaces \cite{nn} to fool $\AC^0$ circuits.  Most recently, Trevisan and Xue \cite{tx} used the insight that pseudorandom restrictions simplify circuits to decision trees as in H{\aa}stad's switching lemma to improve the seed length for depth-$D$ circuits to $\tilde{O}(\log^{D+4}n)$, which remains the best-known generator for $\AC^0$.

For the further restricted class of read-once depth-2 circuits (i.e. CNF or DNF formulas in which every variable appears at most once), Gopalan et al. \cite{gmrtv} constructed a pseudorandom generator generator with seed length $\tilde{O}(\log n)$. 

In this paper, we restrict our attention to read-once $\AC^0$, that is, constant-depth formulas over the de Morgan basis with unbounded fanin.
We continue the approach initiated by Ajtai and Wigderson \cite{aw}, namely that of applying \emph{pseudorandom restrictions} to the circuit to be fooled and incorporate more recent techniques \cite{gmrtv,rsv,svw} into the analysis.


\subsection{Our Results}

Our main result is an improvement upon Trevisan and Xue's $\tilde{O}(\log^{D+4}n)$ seed length \cite{tx} for $\AC^0$ in the special case of read-once $\AC^0$ circuits:

\begin{thm}[Main Result] There is an explicit pseudorandom generator $G: \{0,1\}^{\tilde{O}(\log^{D+1}n)}\to\{0,1\}^n$ fooling read-once $\AC^0$ circuits of depth $D$ on $n$ inputs.\label{thm:main}\end{thm}

In contrast, the probabilistic method implies the existence of an inefficient pseudorandom generator for $\AC^0$ with seed length $O(\log(n/\varepsilon))$ and it is conjectured that efficient generators with matching seed length exist. However, an efficient pseudorandom generator with seed length $o(\log^D (n/\varepsilon))$ would imply stronger circuit lower bounds for $\AC^0$ than are currently known \cite{hastparity}. 
This presents a serious barrier to the construction of pseudorandom generators and our results show that we can match this barrier up to one $\tilde{O}(\log(n/\varepsilon))$ factor in the read-once setting.

\subsection{Our Techniques}

Our pseudorandom generator is that of Gopalan et al.~\cite{gmrtv}, which is also used by Reingold et al.~\cite{rsv} and Steinke et al.~\cite{svw}. Roughly speaking, the generator fixes a carefully chosen fraction of the input bits of a given circuit in a way that approximately preserves the acceptance probability on average. This is applied recursively to fool the circuit using few random bits.


The key to the analysis is \emph{discrete Fourier analysis}:
Fourier analysis has proven highly effective in studying functions on the Boolean hypercube \cite{ODonnell}, 
 finding applications in not just pseudorandomness but also arithmetic combinatorics, circuit complexity, communication complexity, learning theory, and quantum computing. The basic principle is to study a function $F : \{0,1\}^n \to \mathbb{R}$ by expressing it in the Fourier basis, namely $$F(x) = \sum_{s \in \{0,1\}^n} \hat{F}[s]\chi_s(x),$$ where $\chi_s(x) = (-1)^{s\cdot x}$ for $s,x\in\{0,1\}^n$. Of particular relevance to pseudorandomness is the fact that the \emph{Fourier coefficients} $\hat{F}$ can be used to measure the ``complexity'' of $F$. For example, if $\sum_{s \in \{0,1\}^n} |\hat{F}[s]| \leq B$, then $F$ can be $\varepsilon$-fooled by an efficient small-bias generator \cite{nn} with seed length $O(\log(nB/\varepsilon))$. 

Reingold et al.~\cite{rsv} showed that to be fooled by the pseudorandom generator of Gopalan et al.~\cite{gmrtv}, it suffices to satisfy a weaker condition on the Fourier coefficients: we only need to bound the \emph{Fourier growth} --- that is, we must show that $$\forall k \in \{1, 2, \cdots, n\} ~~~~ \sum_{s \in \{0,1\}^n : |s|=k} \left|\hat{F}[s]\right|\le B\cdot c^k$$ for a ``small'' value of $c$ (e.g. $c = \polylog(n)$). By bounding the Fourier growth of read-once, ``permutation'' branching programs, Reingold et al. proved that this generator fools such branching programs; Steinke et al. \cite{svw} then showed a similar bound for all read-once branching programs of width three.

The main contribution of this work is to prove such a Fourier growth bound for the case of read-once $\AC^0$. To our knowledge, while there are known Fourier growth bounds for $\AC^0$ (of a different nature than those we require) due to Linial et al. \cite{lmn} and Impagliazzo and Kabanets \cite{ik} (with implications for the sensitivity and learnability of formulas), and while a Fourier concentration result of Mansour \cite{mansour} was used by De et al. \cite{dett} to show small-bias spaces fool depth-2 circuits, this work is the first to apply Fourier growth bounds to the problem of pseudorandomness against $\AC^0$.


\begin{thm}[Fourier Growth Bound]If $F: \{0,1\}^n\to\{0,1\}$ is computed by a read-once, depth-$D$ circuit, then $$\sum_{s\in \{0,1\}^n : |s|=k}|\hat{F}[s]|\le O(\log^{D-1}(n))^{k}.$$\label{thm:main_fouriergrowth}
\end{thm}

To prove our Fourier growth bound, we induct on depth to show that the Fourier mass at any node of $F$ is either polynomially small or can be bounded in terms of both the acceptance and rejection probabilities at that node.

Theorem~\ref{thm:main_fouriergrowth} together with the analysis of Steinke et al.~\cite{svw} gives a generator with seed length $\tilde{O}(\log^{D+1}(n))$. Roughy speaking, Theorem~\ref{thm:main_fouriergrowth} implies that we can restrict an $\Omega(1/\log^{D-1}(n))$ fraction of inputs via a small-bias space and approximately preserve the acceptance probability (on average). Doing this $O(\log^{D-1}n)\cdot O(\log n)$ times sets all the input bits. Each restriction uses $\tilde{O}(\log n)$ random bits, whence we obtain a pseudorandom generator with seed length $\tilde{O}(\log^{D+1}(n))$.


\subsection{Organization}

In Section~\ref{sec:prelims}, we introduce preliminary definitions and technical tools to be used in our analysis. In Section~\ref{sec:fourier}, we prove our Fourier growth bound. In Section~\ref{sec:gen} we verify that the analysis in \cite{svw} of their pseudorandom restriction generator for branching programs applies to our setting of read-once $\AC^0$ and use the results of the preceding sections to prove that it indeed fools read-once $\AC^0$ circuits.

\section{Preliminaries}
\label{sec:prelims}

\subsection{$\AC^0$ Circuits}

\begin{defn}
A \textbf{read-once, depth-$D$ $\AC^0$ circuit on $n$ inputs} is a Boolean function $F:\{0,1\}^n\to\{0,1\}$ represented by a tree of depth $D$ with $n$ leaves whose nodes either compute the AND or OR of the values computed by their child nodes or the NOT of the value computed by a single child node, and whose output is the value computed by the root of the tree. For a node $f$ of $F$, we say that $f$ is of \textbf{height} $d$ if it is the parent of a node of height $d-1$, and of height 0 if it is a leaf (i.e. an input node). By standard techniques, all the NOT gates can be pushed to the inputs.
\end{defn}

\subsection{Fourier Analysis}

Recall the following basic definitions in Fourier analysis:

\begin{defn}Define the \textbf{characters of $\{0,1\}^n$} to be the maps $\chi_s(x) = (-1)^{x\cdot s}$ for $s\in\{0,1\}^n$, where $x\cdot s$ denotes the bitwise dot product.

For any function $F:\{0,1\}^n\to\mathds{R}$, the \textbf{(discrete) Fourier transform of $F$} is the function $\hat{F}:\{0,1\}^n\to\mathds{R}$ given by $$\hat{F}[s] := \E_{x\sim U}\left[F(x)\cdot\chi_s(x)\right].$$ We call $\hat{F}[s]$ the \textbf{$s$th Fourier coefficient of $F$}, and its \textbf{order} is defined to be $|s|$, the number of nonzero bits in $s$.

The characters form an orthonormal basis for the space of all $F: \{0,1\}^n\to\mathds{R}$. In particular, the \textbf{Fourier expansion of $F$} is $$F(x) = \sum_s\hat{F}[s]\cdot\chi_S(x).$$ The expectation of $F$ under any distribution $X$ can then be written as $$\E_{x\sim X}[F(x)] = \sum_s\hat{F}[s]\cdot\E_{x\sim X}\left[\chi_s(x)\right].$$
\end{defn}

We can now define notions of ``Fourier growth'':

\begin{defn}
The \textbf{Fourier mass at level $k$} of $F:\{0,1\}^n\to\{0,1\}$ is the quantity $$L^k(F) := \sum_{|s|=k}\left|\hat{F}[s]\right|,$$ where for $k<0$ and $k>n$, we say that $L^k(F) = 0$. The \textbf{Fourier mass} of $F$ is merely $\sum_{k\ge 1}L^k(F)$. We also define $L^{\ge k} = \sum_{k'\ge k} L^{k'}(F)$. For any $p\in[0,1]$, the \textbf{$p$-damped Fourier mass} is the quantity $$L_p(F) := \sum_{k>0}p^kL^k(F) = \sum_{s\neq 0}p^{|s|}\cdot\left|\hat{F}[s]\right|.$$
\end{defn}

The motivation for working with $L_p$ is that a bound on $L_p$ yields bounds on each $L^k$.

\begin{lem}
For all $p\in[0,1]$, $$\max_k\left[p^kL^k(F)\right]\le L_p(F)\le n\cdot\max_k\left[p^kL^k(F)\right].$$\label{lem:lp}
\end{lem}

\section{A Fourier Growth Bound}
\label{sec:fourier}

To prove Theorem~\ref{thm:main_fouriergrowth}, we will show that for any function $F$ computed by a read-once $\AC^0$ circuit, $L_p(F)$ can be bounded in terms of the size, depth, and both $\hat{F}[0]$ and $(1-\hat{F}[0])$.

\begin{thm}
	If $F:\{0,1\}^n\to\{0,1\}$ is computed by a read-once, depth-$D$ $\AC^0$ circuit then \begin{equation}L_p(F)\le p\cdot\min(\hat{F}[0],1-\hat{F}[0])\cdot\left(9\log\left(4^D n/\varepsilon\right)\right)^D + \varepsilon.\label{eq:mainbound}\end{equation} for all $\varepsilon\le 1/n$ and $p\le 1/(9\log(4^Dn/\varepsilon))^{D}$.\label{thm:mainbound}
\end{thm}

We will prove the theorem by induction on the depth $D$. The following propositions will allow us to analyze the Fourier growth of formula $F$ in terms of its immediate subformulas (which are at smaller depth).

\begin{prop}
If $F:\{0,1\}^{n_1+n_2}\to\{0,1\}$ is the AND of functions $F_1:\{0,1\}^{n_1}\to\{0,1\}$ and $F_2:\{0,1\}^{n_2}\to\{0,1\}$, then for all $s\in\{0,1\}^{n_1}$ and $t\in\{0,1\}^{n_2}$, $\hat{F}[s\circ t] = \hat{F_1}[s]\cdot\hat{F_2}[t].$\label{prop:fhat}
\end{prop}

\begin{proof}
Because $F = F_1\cdot F_2$, by definition we have that \begin{align*}
\hat{F}[s\circ t] &= \E_{x\circ y\sim  U_{n_1+n_2}}\left[\left(F_1(x)\cdot F_2(y)\right)\chi_{s\circ t}(x\circ y)\right] \\
&= \E_{x\sim  U_{n_1}}\left[F_1(x)\chi_s(x)\right]\cdot \E_{y\sim  U_{n_2}}\left[F_2(y)\chi_t(y)\right] = \hat{F_1}[s]\cdot\hat{F_2}[t],
\end{align*} where in the penultimate equality we use the fact that $\chi_{s\circ t}(x\circ y) = \chi_s(x)\cdot\chi_t(y)$.
\end{proof}

\begin{prop}
If $F: \{0,1\}^{n_1+\cdots+n_m}\to\{0,1\}$ is the AND of functions $F_1:\{0,1\}^{n_1}\to\{0,1\},...,F_m:\{0,1\}^{n_m}\to\{0,1\}$, then $$L_p(F) = \prod^m_{i=1}(L_p(F_i)+\hat{F_i}[0]) - \prod^m_{i=1}\hat{F_i}[0].$$\label{prop:lp}
\end{prop}

\begin{proof}
We will prove this for the case of $m=2$; the proof for general $m$ is entirely analogous.

From Proposition~\ref{prop:fhat}, we have that $L^k(F) = \sum^{n}_{i=0}L^i(F_1)\cdot L^{k-i}(F_2)$ for $k>1$.

Rewrite the left-hand side of the desired equality as \begin{align*}\sum^n_{k=1}p^k L^k(F) &= \left(\sum^n_{k=0} p^k\sum^n_{i=0}L^{i}(F_1)\cdot L^{k-i}(F_2)\right) - L^0(F_1)L^0(F_2)\\
&= \left(\sum^n_{i=0}p^{i}L^{i}(F_1)\right)\cdot\left(\sum^n_{j=0}p^{j}L^{j}(F_2)\right) - L^0(F_1)L^0(F_2) \\
&= \left(L_p(F_1) + L^0(F_1)\right)\cdot\left(L_p(F_2)+L^0(F_2)\right) - L^0(F_1)L^0(F_2)
\end{align*}and we get the desired result because $L^0(F)=\hat{F}[0]$ for all $\{0,1\}$-valued functions $F$.
\end{proof}

We are now ready to prove our Fourier growth bound.

\begin{proof}[Proof of Theorem~\ref{thm:fouriergrowth}]

Base case ($D = 0$): $F$ is a constant, the identity, or the negation of the identity. If $F$ is a constant, then $L_p(F) = 0$. If $F$ is the identity or its negation, then the Fourier expansion of $F$ is either $F(x) = 1/2 - \chi(x)/2$ or $F(x) = 1/2 + \chi(x)/2$, where $\chi(x) = (-1)^x$. For either case, $L_p(F) = p/2$ and $\min(\hat{F}[0],1-\hat{F}[0]) = 1/2$.

Now consider any $F$ computed by a read-once $\AC^0$ circuit of depth $D$ on $n$ inputs. Because both sides of \eqref{eq:mainbound} are invariant under negation of $F$, we can assume without loss of generality that $F$ is the AND of functions $F_1,...,F_k$ computed by circuits of depth $D - 1$ on $n_1,...,n_k$ inputs, respectively; we call these functions the \emph{children} of $F$.

Let $\varepsilon_i = n_i\varepsilon/(4n)$ so that $4^{D-1}n_i/\varepsilon_i = 4^Dn/\varepsilon$ and $\sum\varepsilon_i = \varepsilon/4$. We inductively know that \eqref{eq:mainbound} holds for every $F_i$ and $\varepsilon_i$ so that \begin{equation}L_p(F_i)\le p\cdot\min(\hat{F_i}[0],1-\hat{F_i}[0])\cdot\left(9\log(4^D n/\varepsilon)\right)^{D-1} + \varepsilon_i.\label{eq:indstep}\end{equation}

For the inductive step, roughly, we will show that either the ratio $L_p(F)/\min(\hat{F}[0],1-\hat{F}[0])$ is small, or $L_p(F)<\varepsilon$. Our analysis will be divided into the following three cases: 1) some child of $F$ has very low acceptance probability, 2) the expected number of children $F_i$ of $F$ which output zero under uniformly random assignment to the inputs to $F$ is at most logarithmic, or 3) the expected number of children which output zero is large. In case 1, $\hat{F_i}[0]$ being low for some $i$ inductively implies that $L_p(F_i)$ is low enough that $L_p(F)<\varepsilon$. In case 2, we reduce bounding $L_p(F)$ to bounding $\sum_iL_p(F_i)/\hat{F_i}[0]$, and we again use the inductive hypothesis to argue that this is small. In case 3, we show that $L_p(F)$ is inversely exponential in the expected number of children which output zero and thus that $L_p(F)<\varepsilon$.

\begin{case}
There exists some $i\in[k]$ for which $\hat{F_i}[0] < \varepsilon/4$.
\end{case}

For all $j\in[k]$, by \eqref{eq:indstep}, we have that \begin{align*}L_p(F_j) + \hat{F_j}[0] &\le \hat{F_j}[0]\cdot\left(1 + p\cdot\left(9\log(4^{D} n/\varepsilon)\right)^{D-1}\right) + \varepsilon_j < 3\hat{F_j}[0]/2 + \varepsilon/4,\\
L_p(F_j) + \hat{F_j}[0] &\le \hat{F_j}[0] + (1-\hat{F_j}[0]) \cdot p\cdot\left(9\log(4^{D} n/\varepsilon)\right)^{D-1} + \varepsilon_j < 1 + \varepsilon/4.
\end{align*} 
Since $\hat{F_i}[0] < \varepsilon/4$, the former inequality gives $L_p(F_i) + \hat{F_i}[0] < 5\varepsilon/8$.
Moreover, $L_p(F_j) + \hat{F_j}[0]< 1 + \varepsilon/4$ for all $j\neq i$. Thus, by Proposition~\ref{prop:lp}, we have $$L_p(F) \leq \prod_{j=1}^k (L_p(F_j) + \hat{F_j}[0]) < \frac{5}{8} \varepsilon \cdot (1 + \varepsilon/4)^{k-1} \leq \varepsilon,$$ as $\varepsilon \leq 1/k$.

\begin{case}
$\hat{F_i}[0]\ge \varepsilon/4$ for all $i\in[k]$ and $\sum_i(1-\hat{F_i}[0]) < 2\log(4^Dn/\varepsilon)$.
\end{case}

We can rewrite $L_p(F)$ as \begin{align*}L_p(F) &= \left(\prod_i\hat{F_i}[0]\right)\cdot\left(\prod_i\left(\frac{L_p(F_i)}{\hat{F_i}[0]} + 1\right) - 1\right) \\
&\le \hat{F}[0]\cdot\left(\exp\left(\sum_i\frac{L_p(F_i)}{\hat{F_i}[0]}\right) - 1\right).\numberthis\label{eq:exp}\end{align*} 
Now we must simply upper bound $\sum_i L_p(F_i)/\hat{F_i}[0] $. Since $\min(x,1-x)\le 2x(1-x)$ for any $x\in[0,1]$, by \eqref{eq:indstep} we have \begin{align*}\sum_i\frac{L_p(F_i)}{\hat{F_i}[0]}&\le \sum_i 2p\cdot(1-\hat{F_i}[0])\cdot\left(9\log\left(4^{D}n/\varepsilon\right)\right)^{D-1} + \sum_i\varepsilon_i/\hat{F_i}[0] \numberthis \label{eq:ratio}\\
&\le p\cdot(4/9)\cdot(9\log(4^Dn/\varepsilon))^D + 1 < 2,
\end{align*} where the penultimate inequality follows from the hypotheses of Case 2. Applying the inequality $e^x - 1\le 4x$ for $x\le 2$ to \eqref{eq:exp} gives \begin{equation}L_p(F)\le\hat{F}[0]\cdot\left(4\sum_i\frac{L_p(F_i)}{\hat{F_i}[0]}\right).\label{eq:fork}\end{equation}

Suppose $\hat{F}[0]> 1/2$. Then because $e^{-2x}\le 1 - x$ for $0 \leq x\le 1/2$, we have $$\exp\left(-2(1-\hat{F}[0])\right)\le \hat{F}[0] = \prod_i(1-(1-\hat{F_i}[0]))\le \exp\left(-\sum_i(1-\hat{F_i}[0])\right)$$ and thus $\sum_i(1-\hat{F_i}[0])\le 2(1-\hat{F}[0])$. By \eqref{eq:fork} and \eqref{eq:ratio}, we have \begin{align*}L_p(F) &\le 8\hat{F}[0]\cdot p\cdot (9\log(4^Dn/\varepsilon))^{D-1}\cdot \sum_i(1-\hat{F_i}[0]) + 4\sum_i \varepsilon_i\cdot\frac{\hat{F}[0]}{\hat{F_i}[0]} \\
&\le 16p\cdot(9\log(4^Dn/\varepsilon))^{D-1}\cdot(1-\hat{F}[0]) + \varepsilon\end{align*} as desired, where in the latter inequality we used the fact that $\hat{F}[0]/\hat{F_i}[0]\le 1$ for all $i\in[k]$.

Now suppose $\hat{F}[0] \le 1/2$. Then by \eqref{eq:ratio}, we can rewrite \eqref{eq:fork} as $$L_p(F)\le\hat{F}[0]\cdot\left(p\cdot(9\log(4^Dn/\varepsilon))^D + 4\sum_i\varepsilon_i/\hat{F_i}[0]\right)< p\cdot\hat{F}[0]\cdot(9\log(4^D n/\varepsilon))^D + \varepsilon.$$

\begin{case}
$\hat{F_i}[0]\ge \varepsilon/4$ for all $i\in[k]$ and $\sum_i(1-\hat{F_i}[0])\ge 2\log(4^Dn/\varepsilon)$.
\end{case}

By \eqref{eq:indstep}, \begin{align*}\prod_{i}(L_p(F_i)+\hat{F_i}[0]) &\le \prod_{i}\left(\hat{F_i}[0] + p(1-\hat{F_i}[0])(9\log(4^{D}n/\varepsilon))^{D-1} + \varepsilon_i\right) \\
&= \prod_{i}\left(1 - (1-\hat{F_i}[0])\left(1-p(9\log(4^{D}n/\varepsilon))^{D-1}\right) + \varepsilon_i\right) \\
&\le 1/\exp\left(\sum_i\left((1-\hat{F_i}[0])\left(1-p(9\log(4^{D}n/\varepsilon))^{D-1}\right) - \varepsilon_i\right)\right).\end{align*} But because $p\le 1/(9\log(4^Dn/\varepsilon))^{D}$, $p(9\log(4^Dn/\varepsilon))^{D-1} < 0.1$, so \begin{align*}\sum_i\left((1-\hat{F_i}[0])\left(1-p(9\log(4^{D}n/\varepsilon))^{D-1}\right) - \varepsilon_i\right) &> 0.9\sum_i(1 - \hat{F_i}[0]) - \varepsilon/4\\
&\ge 1.8\log(4^D n/\varepsilon) - \varepsilon/4 \\
&> \log(4^D n/\varepsilon) > \log(1/\varepsilon),
\end{align*} so we conclude that $\prod_i(L_p(F_i)+\hat{F_i}[0])<\varepsilon$.
\end{proof}

\begin{cor}
If $F: \{0,1\}^n\to\{0,1\}$ is computed by a read-once $\AC^0$circuit of depth $D = O(1)$, then $L_p(F)\le O(1)$ for $p\le 1/(9\log(4^Dn/\varepsilon))^{D}$, so in particular, by Lemma~\ref{lem:lp}, $$L^k(F)\le O(\log^{D-1}n)^k$$ for all $k$.\label{thm:fouriergrowth}
\end{cor}

\begin{proof}
	As before, say that $F$ is the AND of some $F_1,...,F_k$. If we apply Theorem~\ref{thm:main} to each $F_i$ with $p = 1/(9\log(4^{D-1}n/\varepsilon))^{D-1}$ to get $$L_p(F_i) + \hat{F_i}[0]\le \min(\hat{F_i}[0],1-\hat{F_i}[0]) + \varepsilon + \hat{F_i}[0]\le 1 + \varepsilon.$$ Therefore, by Proposition~\ref{prop:lp}, $L_p(F)\le (1+\varepsilon)^k$. In particular, for $D = O(1)$ and $\varepsilon = 1/n$, $L_p(F) < O(1)$ as desired.
\end{proof}

Note that the proof of our Fourier growth bound amounts to inductively showing in Theorem~\ref{thm:mainbound} that for fixed $p = 1/(9\log(4^Dn/\varepsilon))^{D-1}$, \eqref{eq:mainbound} holds for every descendant of the root, and then concluding in the proof of the above corollary that at the root, $L_p(F)$ is small because $L_p(F_i)$ is small for all children $F_i$.

The reason the analysis for the root of $F$ differs from that for its descendants is that we cannot strengthen Theorem~\eqref{thm:mainbound} to show $$L_p(F)\le p\cdot\min(\hat{F}[0],1-\hat{F}[0])\cdot\left(9\log\left(4^D n/\varepsilon\right)\right)^{D-1} + \varepsilon.$$ for all $p\le 1/O(\log(n/\varepsilon))^{D-1}$. For example, when $D = 1$, this would say that for all sufficiently small $p$, we have $L_p(F)\le O(p\cdot\min\{\hat{F}[0],1-\hat{F}[0]\}) + \varepsilon$. This is false for $F = \bigwedge^k_{i=1} X_i$ when $k = \log(1/\varepsilon)$ because then $$L_p(F) = \left(\frac{p}{2} + \frac{1}{2}\right)^k - \frac{1}{2^k} = \frac{1}{2^k}e^{\Omega(kp)},$$ but $O(p\cdot\min\{\hat{F}[0],1-\hat{F}[0]\}) + \varepsilon = \frac{O(p) + 1}{2^k} < L_p(F)$.

Furthermore, as discussed in \cite{rsv}, Fourier growth bounds are related to the Coin Theorem of Brody and Verbin \cite{bv}. They proved that for a read-once, width-$(D+1)$ branching program $F$ to distinguish the distribution $X\in\{0,1\}^n$ of $n$ independent samples from a coin with bias $p\in[-1,1]$ from the uniform distribution, $|p|$ must be at least $\Omega(\log^{1-D}n)$. Specifically, they show that for any such $F$, $|\E_X[F(X)] - \E_U[F(U)]|\le O(|p|(\log n)^{D-1})$. In Fourier analytic terms, \begin{equation}\left|\E_X[F(X)] - \E_U[F(U)]\right| = \left|\sum_{s\neq 0}\hat{F}[s]\cdot p^{|s|}\right|,\label{eq:bv}\end{equation} which is simply $L_p$ without absolute values. Read-once $\AC^0$ circuits of depth $D$ can be simulated by read-once, width-$(D+1)$ branching programs, and just as Brody and Verbin show that \eqref{eq:bv} is small for $p = 1/O(\log^{D-1}n)$ for read-once branching programs, Corollary~\ref{thm:fouriergrowth} shows that $L_p$ is small for this setting of $p$ for read-once $\AC^0$ circuits.

Moreover, by using the recursive tribes formula, Brody and Verbin show that their bound is essentially tight in the choice of $p$, implying that our bound is tight as well.

\section{The Pseudorandom Generator}
\label{sec:gen}

In this section, we will show that the pseudorandom restriction generator of \cite{svw} can be used to fool read-once $\AC^0$ circuits. Their result deals with fooling families of branching programs, so before recalling this result, we will define the relevant terminology.

\subsection{Branching Programs}

\begin{defn}
	A length-$n$, width-$w$ \textbf{branching program} is a function $B: \{0,1\}^n\times[w]\to[w]$ which takes a start state $u\in[w]$ and an input string $x\in\{0,1\}^n$ and outputs a final state $B[x](u)$.

	We will think of $B$ as having a fixed \textbf{start state} and \textbf{accept state}, both of which for convenience we will denote by the index 1. Then $B$ \textbf{accepts} $x\in\{0,1\}^n$ if $B[x](1) = 1$, and we say that $B$ \textbf{computes the function} $F: \{0,1\}^n\to\{0,1\}$ if $F(x) = 1$ if and only if $B[x](1) = 1$.
\end{defn}

A branching program reads a single bit of the input at a time (rather than reading $x$ all at once) and only keeps track of the state in $[w]$ at each step. We enforce this by requiring the program to be composed of smaller programs as follows.

\begin{defn}
	If $B$ and $B'$ are width-$w$ branching programs of length $n$ and $n'$ respectively, then the \textbf{concatenation} $B\circ B': \{0,1\}^{n+n'}\times[w]\to[w]$ of $B$ and $B'$ is the length-$(n+n')$, width-$w$ program defined by $$(B\circ B')[x\circ x'](u) := B'[x'](B[x](u)).$$ That is, first $B\circ B'$ runs $B$ on the first part of the input, then the start state of $B'$ is set to the final state of $B$, and then $B\circ B'$ runs $B'$ on the rest of the input.
\end{defn}

\begin{defn}
	A length-$n$, width-$w$ \textbf{ordered branching program} is a read-once program $B$ that can be written as $B = B_1\circ\cdots\circ B_n$ where each $B_i$ is a length-1, width-$w$ program. We will refer to $B_i$ as the $i$th \textbf{layer} of $B$, and $B_{i\cdots j}:= B_i\circ\cdots\circ B_j$ will denote the \textbf{subprogram} of $B$ from layer $i$ to $j$.
\end{defn}

\begin{figure}
	\centering
	\includegraphics[width=0.6\textwidth]{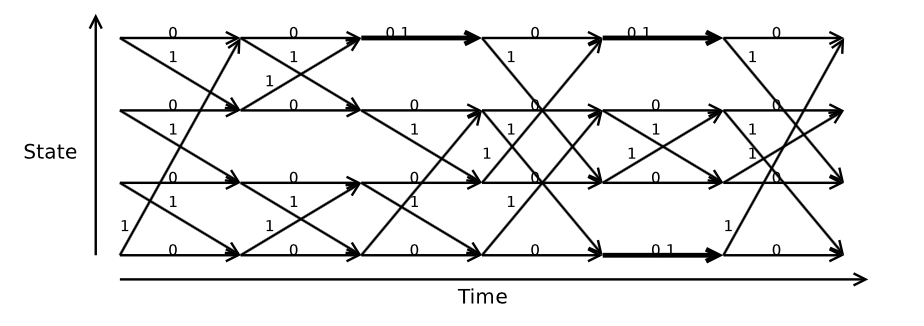}
	\caption{An example illustration of a length-6, width-4 branching program \cite{svw}}
\end{figure}

A length-$n$, width-$w$ ordered branching program can also be regarded as a directed acyclic graph. The vertices are arranged into $n+1$ layers each of size $w$. The edges connect vertices in adjacent layers; in particular, for each layer $i$, each vertex $u$ in layer $i$, and each $b\in\{0,1\}$, there is an edge labeled $b$ from $u$ to vertex $B_{i}[b](u)$ in layer $i+1$.

We use the following notational conventions when referring to layers of a length-$n$ branching program. There is a distinction between layers of edges and layers of vertices: the former are the length-1 subprograms $B_i$ defined above and are numbered from 1 to $n$, while the latter are the states between the $B_i$s and are numbered from 0 to $n$. The edges in $B_i$ go from vertices in layer $i-1$ to vertices in layer $i$.

Lastly, as mentioned in the introduction, the pseudorandom generator we will use makes use of pseudorandom restrictions. We formalize the notion of restrictions to Boolean functions.

\begin{defn}
	For $t,x\in\{0,1\}^n$, and $F: \{0,1\}^n\to\{0,1\}$ the \textbf{restriction of $F$ to $t$ using $x$}, denoted $F|_{\overbar{t}\leftarrow x}$, is the function obtained by setting the inputs indexed by the zero bits of $t$ to the corresponding bits of $x$ and leaving the inputs indexed by the nonzero bits of $t$ free. Formally, $$F|_{\overbar{t}\leftarrow x}(y) = F(\text{Select}(t,y,x)),$$ where $$\text{Select}(t,y,x)_i = \begin{cases}
		y_i & t_i = 1\\
		x_i & t_i = 0
	\end{cases}.$$ We can define restrictions $B|_{\overbar{t}\leftarrow x}$ of branching programs $B: \{0,1\}^n\times[w]\to[w]$ analogously.
\end{defn}

\subsection{Closure Under Restrictions, Subprograms, and Permutations}

We now state the result of \cite{svw} on pseudorandomness for branching programs and show that it can be applied to our setting.

\begin{thm}[\cite{svw}, Theorem 5.1]
	Let $\mathcal{C}$ be a family of ordered branching programs of length at most $n$ and width at most $w$ that is closed under taking restrictions, taking subprograms, and permuting layers -- that is, if $B\in\mathcal{C}$ computes some function $F:\{0,1\}^n\to\{0,1\}$, then $B|_{t\to x}\in\mathcal{C}$ for all $t,x\in\{0,1\}^n$, $B_{i\cdots j}\in\mathcal{C}$ for all $1\le i<j\le n$, and $\pi B, B\pi\in\mathcal{C}$ for all permutations $\pi: [w]\to[w]$ where $(\pi B)[x](w) = B[x](\pi(w))$ and $(B\pi)[x](w) = \pi(B[x](w))$. Suppose that, for all $k\in[n]$ and all $F$ computed by some $B\in\mathcal{C}$, we have $L^k(F)\le ab^k$, where $b\ge 2$.

	Then for $\varepsilon> 0$, there exists a pseudorandom generator $G_{a,b,n,\varepsilon}: \{0,1\}^{s_{a,b,n,\varepsilon}}\to\{0,1\}^n$ with seed length $s_{a,b,n,\varepsilon} = O\left(b\cdot\log(b)\cdot\log(n)\cdot\log\left(\frac{abw^2n}{\varepsilon}\right)\right)$ such that, for any $F$ computed by some $B\in\mathcal{C}$, $$\left|\E_{U_{s_a,b,n,\varepsilon}}[F(G_{a,b,n,\varepsilon}(U_{s_{a,b,n,\varepsilon}})) - \E_U[F(U)]\right|\le\varepsilon.$$ Moreover, $G_{a,b,n,\varepsilon}$ can be computed in space $O(s_{a,b,n,\varepsilon})$.\label{thm:svw}
\end{thm}

Note that the statement above differs slightly from the statement in \cite{svw}; in particular, the seed length $s_{a,b,n,\varepsilon}$ above is related to their seed length $t_{a,b,n,\varepsilon}$ by $s_{a,b,n,\varepsilon} = t_{wa,b,n,\varepsilon}$. The reason is that in \cite{svw}, branching programs are regarded as matrix-valued functions $B: \{0,1\}^n\to\{0,1\}^{w\times w}$ where $B[x]_{(u,v)} = 1$ if and only if $B[x](u) = v$, whereas we are concerned only with the Boolean functions computed by branching programs.

In the theorem stated in \cite{svw}, the hypothesis was that $L^k(B)\le ab^k$, where $L^k(B)$ is defined in terms of the \emph{matrix-valued Fourier transform} and the subordinate $L^2$ matrix norm $\norm{\cdot}_2$. In general, if $M$ is a $w\times w$ matrix whose entries are each bounded in absolute value by $C$, then $\norm{M}_2\le w\cdot C$. Therefore, $L^k(B)\le w\cdot\max_{u,v\in[w]}L^k(F_{u,v})$, where $F_{u,v}$ is the function computed by $B$ if we use $u$ as the start state and $v$ as the accept state. But since the family $\mathcal{C}$ is closed under permuting layers, we have a bound on $L^k(F_{u,v})$ for all $u,v$.

To apply their construction to our setting, we need to show that every function $F$ computed by a read-once $\AC^0$ circuit is computed by some branching program $B$ whose restrictions and subprograms can be simulated by read-once $\AC^0$ circuits.

Firstly, given a branching program $B$, vertex layers $i,j\in[n]$, and states $d_1,d_2\in[w]$, define $B^{d_1,d_2}_{i\cdots j}: \{0,1\}^{j-i+1}\to\{0,1\}$ by $$B^{d_1,d_2}_{i\cdots j}(x) = \mathds{I}[B_{i\cdots j}[x](d_1) = d_2].$$

Now define the class $\mathcal{C}$ to be the set of ordered, length-$n$, width-$D+1$ branching programs $B$ on variable sets $V(B)\subseteq[n]$ such that for all $i,j\in V(B)$ and $d_1,d_2\in[D+1]$, $B^{d_1,d_2}_{i\cdots j}$ is computed by an $\AC^0$ read-once formula of depth $D$.

\begin{prop}
	If $F:\{0,1\}^n\to\{0,1\}$ is computed by a read-once, depth-$D$ $\AC^0$ circuit, then $F$ is also computed by an ordered, length-$n$, width-$(D+1)$ branching program $B\in\mathcal{C}$.
\end{prop}

\begin{proof}
	We will induct on depth. The claim is trivially true for $D = 0$ in which $F$ can only be a constant, the identity, or the negation of the identity.

	Now consider any $F$ computed by a read-once $\AC^0$ circuit of depth $D$ on $n$ inputs. Assume without loss of generality that $F$ is the AND of functions $F_1,...,F_k$ computed by circuits of depth $D-1$ on $n_1,...,n_k$ inputs respectively (the argument for the case where $F$ is an OR of functions is completely analogous).

	Inductively, we have ordered branching programs $B^1,...,B^k\in\mathcal{C}$ of width $D$ on $n_1,...,n_k$ inputs which compute $F_1,...,F_k$ respectively. To construct the desired branching program $B$ for $F$, we essentially concatenate the $B^1,..,B^k$ and, for each $i\in[k-1]$, connect the accept state in the last layer of $B^i$ to the start state in the first layer of $B^{i+1}$ and connect the non-accept states in the last layer of $B^i$ to a non-accept state in the last layer of $B^k$.
	
	Formally, for each $B^i$ define $B'^i$ to be the width-$(D+1)$ program given by introducing an extra state \texttt{reject} to each layer of vertices and rearranging the edges in the last layer that do not lead to the accept state to lead to the \texttt{reject} state instead. Specifically, define $B'^i = B'^i_1\circ\cdots\circ B'^i_{n_i}$ for length-1, width-$(D+1)$ programs $\{B'^i_j\}$ as follows. For $x\in\{0,1\}$, $u\in[D+1]$, and $m\in[n_i]$, $$B'^{i}_m[x](u) = \begin{cases} 
      \texttt{reject} & u = \texttt{reject}, \ \text{or} \\
      & m = n_i \ \text{and} \ B^i_m[x](u)\neq 1 \\
      B^i_m[x](u) & \text{otherwise} 
   \end{cases}$$

	Now define $B$ to be $B'^1\circ\cdots B'^k$. $F$ is satisfied if and only if each of the $F_i$ is satisfied. By construction, each $B'^{i+1}$ can only end on 1 or \texttt{reject}, and it ends on 1 if and only if in the computation of $B$, $B'^{i+1}$ started in state 1 and $F_{i+1}$ is satisfied. But the former holds if and only if $B'^i$ ended in state 1, so we conclude that $B$ ends on 1 if and only $B'^i$ ends on 1 for all $i$, which happens if and only if $F_i$ outputs 1 for all $i$. Therefore, $B$ computes $F$.

	It just remains to check that every $B^{d_1,d_2}_{i\cdots j}$ can also be computed by a read-once $\AC^0$ circuit. If the first and last layers of $S$ both lie in a single $B'^m$, then we're done by the inductive hypothesis on $F_m$. Otherwise, suppose $S$ starts at state $d_1$ of the the $i_1$th layer of $B'^{j_1}$ and ends at state $d_2$ of the $i_2$nd layer of $B'^{j_2}$. By the inductive hypothesis on $F_{j_1}$ and $F_{j_2}$, the subprograms $(B'^{j_1})^{d_1,1}_{i\cdots n_{j_1}}$ and $(B'^{j_2})^{1,d_2}_{1\cdot n_{j_2}}$ are computed by read-once $\AC^0$ circuits of depth $D-1$, call them $G$ and $H$. Then the function that the subprogram $S$ computes is also computed by the depth-$D$ circuit $$G\wedge F_{\ell+1}\wedge\cdots\wedge F_{m-1}\wedge H.$$
\end{proof}

It is fairly immediate that $\mathcal{C}$ is closed under taking restrictions, taking subprograms, and permuting layers. Certainly if $B\in\mathcal{C}$, then $B_{i\cdots j}\in\mathcal{C}$. Furthermore, if each $B^{d_1,d_2}_{i\cdots j}$ is computed by a read-once $\AC^0$ circuit $F^{d_1,d_2}_{i\cdots j}$, then $(B|_{\overbar{t}\leftarrow x})^{d_1,d_2}_{i\cdots j}$ is computed by $\left(F^{d_1,d_2}_{i,j}\right)\big|_{\overbar{t}\leftarrow x}$. Likewise, $(\pi B)^{d_1,d_2}_{i\cdots j}$ and $(B\pi)^{d_1,d_2}_{i\cdots j}$ are computed by $F^{\pi(d_1),d_2}_{i,j}$ and $F^{d_1,\pi(d_2)}_{i,j}$ respectively.

We can now take the family of ordered branching programs in the statement of Theorem~\ref{thm:svw} to be this family $\mathcal{C}$. By our Fourier growth bound in Corollary~\ref{thm:fouriergrowth}, we obtain a pseudorandom generator for read-once $\AC^0$.

\begin{cor}
For every $n\in\mathds{N}$, $\epsilon>0$, there exists a pseudorandom generator $G: \{0,1\}^{s_{n,\epsilon}}\to\{0,1\}^n$ for $s_{n,\epsilon} = \tilde{O}(\log^{D}n\cdot\log(n/\varepsilon))$ that $\varepsilon$-fools any function $F$ computed by a read-once $\AC^0$ circuit of depth $D$ on $n$ inputs.
\end{cor}

\section{Future Work}

Motivated by the analysis of \cite{gmrtv} in the case of read-once CNFs $F$, we see two directions for improvement upon the current seed length of $\tilde{O}(\log^{D+1}(n))$.

Firstly, we could try relaxing our notion of Fourier growth: rather than bounding $L^k(F)$, it suffices to bound $L^k(G)$ \emph{where $G$ approximates $F$}:

\begin{prop}[\cite{dett}, Proposition 2.6]
	Let $F, F_+, F_-: \{0,1\}^n\to\mathds{R}$ satisfy $F_-(x)\le F(x)\le F_+(x)$ for all $x$ and $\E_U[f_+(U) - f_-(U)]\le\delta$. Then if $X$ is an $\varepsilon$-biased distribution, $$\left|\E_X[F(X)] - \E_U[F(U)]\right|\le\delta + \varepsilon\cdot\max\{L(F_+),L(F_-)\}.$$
\end{prop}

The functions $F_+$ and $F_-$ are called \textbf{$\delta$-sandwiching approximators} for $F$. Gopalan et al. \cite{gmrtv} used the results of \cite{dett} to construct sandwiching approximators with low $L_1$-norm for read-once CNFs, and these approximators allowed them to set a constant fraction of the bits at each level of recursion $(p = \Omega(1))$, whereas the generator we use only sets a $1/O(\log n)$ fraction at each level (when $D = 2$). We would thus like to similarly exploit sandwiching approximators for arbitrary read-once $\AC^0$ circuits to improve the seed length of the generator.

Additionally, Gopalan et al. \cite{gmrtv} showed that after each round of pseudorandomly restricting a constant fraction of the input bits, $F$ shrinks from $m$ to $m^{1-\Omega(1)}$ clauses, so after only $O(\log\log n)$ (rather than $O(\log n)$) steps, the resulting CNF is sufficiently small with high probability that it can be fooled directly by a small-bias space\footnote{More precisely, they show that $F$ has sandwiching approximators that shrink with high probability under the pseudorandom restrictions.}.

We would also like to argue that arbitrary read-once $\AC^0$ circuits shrink well under pseudorandom restrictions. At least in the case of \emph{truly random} restrictions, as we show in Appendix~\ref{app:shrink}, it is true that read-once $\AC^0$ circuits with all but $1/\polylog(n)$ of the input bits restricted will shrink with high probability to size $\polylog(n)$, which gives hope that our seed length can be reduced at least to $\tilde{O}(\log^Dn)$. That said, it is not immediately clear to the authors how to modify the argument to handle \emph{pseudorandom} restrictions.

\appendix

\section{Random Restrictions Simplify Circuits}
\label{app:shrink}

We prove that any read-once $\AC^0$ circuit is approximated by read-once $\AC^0$ circuits which shrink to polylogarithmic size with high probability under a truly random restriction of sufficiently many bits.

First, we make precise the distribution from which we are sampling our restrictions.

\begin{defn}
	A distribution $T$ on $\{0,1\}^n$ is \textbf{$p$-regular} if each bit is independently set to 1 with probability $p$.
\end{defn}

The restrictions $F|_{\overbar{t}\leftarrow x}$ we will be considering are such that $t\sim T$ and $x\sim  U$ for $T$ a $p$-regular distribution and $ U$ the uniform distribution.

\begin{thm}
For $\varepsilon = 1/\poly(n)$, let $F: \{0,1\}^n\to\{0,1\}$ be computed by a read-once, depth-$D$ circuit. Let $T$ be a $p$-regular distribution for $p = 1/O(\log^{D-1}n)$ and $U$ the uniform distribution on $\{0,1\}^n$.

Then $F$ has $O(n\sqrt{\varepsilon})$-sandwiching approximators $F_{\ell}$ and $F_u$ computed by read-once $\AC^0$ circuits of depth $D$ such that $F_{\ell}|_{\overbar{t}\leftarrow x}$ and $F_u|_{\overbar{t}\leftarrow x}$ are of size at most $\tilde{O}(\log^D n)$ with probability at least $1 - 2\varepsilon$ over the choice of $x\sim U$, $t\sim T$.
\label{thm:concshrink}
\end{thm}

For the rest of this section, we will assume without loss of generality that the circuits we are dealing with consist solely of NAND gates, potentially with some NOT gates over the inputs. Indeed, any AND gate can be replaced with a negated NAND gate, and any OR of nodes can be replaced with the NAND of the negations of those nodes. By standard techniques, all the negations can be moved to lie directly above the inputs.

\subsection{Collapse Probability}
\label{sec:regularshrink}

To prove Theorem~\ref{thm:concshrink}, we will first prove that by Theorem~\ref{thm:mainbound}, the probability that a read-once $\AC^0$ circuit does not collapse to a constant under $p$-regular restriction is small relative to its acceptance and rejection probabilities. This lemma will then allow us to prove Theorem~\ref{thm:concshrink} in the last subsection by generalizing the arguments of \cite[Lemma 7.3]{gmrtv} and \cite[Corollary 7.4]{gmrtv} from depth-2 circuits to arbitrary constant depth.

\begin{lem}
Let $F: \{0,1\}^n\to\{0,1\}$ be computed by a read-once $\AC^0$ circuit of depth $D$. For any $\varepsilon < 1/n$, if $p\le 1/(9\log(4^Dn/\epsilon))^D$ and $T$ is a $p$-regular distribution on $\{0,1\}^n$, then $$\Pr[F|_{\overbar{t}\leftarrow x} \ \text{is nonconstant}] \le 2p\cdot\min\left(\hat{F}[0], 1-\hat{F}[0]\right)\cdot\left(9\log(4^D n/\epsilon)\right)^D + 2\epsilon.$$\label{lem:fourway}
\end{lem}

\begin{proof}
	Without loss of generality, we can assume that $F$ is monotone: if we have another $F'$ given by adding NOT gates above some of the inputs, then because each bit is set to 0 or 1 with equal probability, $F|_{\overbar{t}\leftarrow x}$ and $F'|_{\overbar{t}\leftarrow x}$ have the same probability of remaining nonconstant.

	By monotonicity, $F|_{\overbar{t}\leftarrow x}$ is nonconstant if and only if $(F|_{\overbar{t}\leftarrow x})(\overbar{0}) = (F|_{\overbar{t}\leftarrow x})(\overbar{1})$, where $\overbar{0}$ and $\overbar{1}$ denote the strings of $n$ repeated 0's and repeated 1's respectively.

	But $$\E_{x\sim U, t\sim T}\left[(F|_{\overbar{t}\leftarrow x})(\overbar{1}) - (F|_{\overbar{t}\leftarrow x})(\overbar{0})\right] = \left|\E_X[F(X)] - \E_Y[F(Y)]\right|,$$ where $X$ and $Y$ are the distributions of $n$ independent samples from a coin with bias $p$ and $-p$, respectively. By \eqref{eq:bv} and the triangle inequality, $$\left|\E_X[F(X)] - \E_Y[F(Y)]\right| \le\left|\sum_{s\neq 0}\hat{F}[s]p^{|s|}\right| + \left|\sum_{s\neq 0}\hat{F}[s](-p)^{|s|}\right| \le 2L_p(F),$$ so we're done by Theorem~\ref{thm:mainbound}.
\end{proof}

\subsection{Concentrated Shrinkage}

\begin{lem}
For $\varepsilon = 1/\poly(n)$, let $F: \{0,1\}^n\to\{0,1\}$ be computed by a read-once, depth-$D$ circuit such that for each node $f$, $1-\hat{f}[0]\ge\varepsilon$. If $T$ is a $p$-regular distribution and $ U$ is the uniform distribution, then $F|_{\overbar{t}\leftarrow x}$ is of size $\tilde{O}(\log^D n)$ with probability at least $1 - \varepsilon$ over the choice of $x\sim  U$, $t\sim T$.
\label{lem:gmrtv}
\end{lem}

\begin{proof}
Our claim is that each remaining node in $F|_{\overbar{t}\leftarrow x}$ fails to have fan-in at most $\tilde{O}(\log n)$ with probability at most $\varepsilon/(nD)$ so that by the union bound, $F|_{\overbar{t}\leftarrow x}$ fails to have the desired size with probability at most $\varepsilon$.

Fix some node $f$ of $F$, and partition its children into chunks $C_0,...,C_m$ where $C_i$ is the set of all children $c$ for which $2^i\le 1-\hat{c}[0]/\varepsilon \le 2^{i+1}$. Note that $m\le O(\log n)$ because $\varepsilon = 1/\poly(n)$. Let $\varepsilon_i = \prod_{c\in C_i}\hat{c}[0]$ so that $\prod_i\varepsilon_i = 1-\hat{f}[0]\ge\varepsilon.$ For any $i$, $\varepsilon_i \le (1-2^i\varepsilon)^{|C_i|}$ so that \begin{equation}|C_i|\le\frac{1}{2^i\varepsilon}\log(1/\varepsilon_i)\label{eq:sizechild}\end{equation}

Denote the nodes of $C_i$ by $c^i_1,...,c^i_{|C_i|}$, and let $Y^i_j$ be the indicator variable equal to 1 if $c^i_j$ survives in $F|_{\overbar{t}\leftarrow x}$ (i.e. does not collapse to a constant), and 0 otherwise. Note that \begin{equation}\Pr(Y^i_j = 1)\le 2p\cdot(1-\hat{c^i_j}[0])\cdot\left(9\log(4^{d-1}n/\epsilon)\right)^{d-1} + 2\epsilon < (2^{i+1}+2)\varepsilon\label{eq:pryj}\end{equation} where the penultimate inequality follows by Lemma~\ref{lem:fourway}. We want to show that for each $i$, $\sum_j Y^i_j$ is small with high probability.

Let $M\in\mathbb{Z}$ and $k<M$ be some parameters which we will determine later, and let $S_k(Y^i_1,...,Y^i_{|C_i|})$ denote the $k$th symmetric polynomial in the variables $Y^i_j$.\footnote{The $k$th symmetric polynomial in $x_1,...,x_n$ is defined to be $\sum_{1\le i_1<\cdots<i_k\le n}\prod^k_{j=1}x_{i_j}$.} It follows that $$\Pr\left[\sum_jY^i_j > M\right]\cdot\binom{M}{k}\le\E[S_k(Y^i_1,...,Y^i_{|C_i|})]\le\binom{|C_i|}{k}\cdot\left((2^{i+1}+2)\varepsilon\right)^k,$$ where the former inequality holds by noting that if more than $M$ of the $Y^i_j$ are 1, then there are at least $\binom{M}{k}$ terms equal to 1 in $S_k(Y^i_1,...,Y^i_{|C_i|})$, and the latter inequality holds by \eqref{eq:pryj} and independence. Stirling's approximation and \eqref{eq:sizechild} give that $$\Pr\left[\sum_jY^i_j>M\right]\le\left(\frac{|C_i|e}{M}\cdot (2^{i+1}+2)\varepsilon\right)^k\le \left(\frac{3e\log(1/\varepsilon_i)}{M}\right)^k.$$

Now if $\varepsilon/(mnD)>1/n^{c}$ for some constant $c$, take $M$ to be $3e\log\log(n)^{c'}\log(1/\varepsilon_i)$ and $k$ to be $\log\log(n)^{c'}$ for a large enough constant $c'$ that $(\log\log(n)^{c'})^{\log\log(n)^{c'}} > n^c$ and $\Pr\left[\sum_jY^i_j>M\right]<\varepsilon/(mnD)$. A union bound over the $m$ choices of $i$ and the at most $nD$ choices of node $f$ gives the desired bound on probability that fan-in at $f$ is at most $$\sum_i3e\log\log(n)^{c''}\log(1/\varepsilon_i) = O\left(\log\log(n)^{c''}\right)\sum_i\log(1/\varepsilon_i)\le \tilde{O}(\log n),$$ where the last equality follows because $\prod_i\varepsilon_i\ge\varepsilon = 1/\poly(n)$.
\end{proof}

We now drop the assumption that rejection probability is not too small in order to prove Theorem~\ref{thm:concshrink}.

\begin{proof}[Proof of Theorem~\ref{thm:concshrink}]
If $F$ has the property that $1-\hat{f}[0]\ge\varepsilon$ for every node $f$, then by Lemma~\ref{lem:gmrtv}, we can take $F_{\ell}$ and $F_u$ to be $F$ itself. Otherwise, we will show how to modify $F$ to obtain sandwiching formulas with this property.

Let $L(G)$ denote the number of leaves of a formula $G$. We inductively show that each node $f$ of depth $d$ has $O(L(f)\sqrt{\varepsilon})$-sandwiching formulas $f_{\ell}$ and $f_u$ such that i) if $f_{\ell}$ (resp. $f_u$) is not a constant, then $\varepsilon\le 1-\hat{f_{\ell}}[0]\le 1-\varepsilon$ (resp. $\varepsilon\le 1-\hat{f_u}[0]\le 1-\varepsilon$), ii) $L(f_u),L(f_{\ell})\le L(f)$.

This is certainly true for the leaves of $F$. Now fix a node $f$ of depth $d$; for each $c\in c(f)$, we have sandwiching $c_{\ell}$ and $c_u$ satisfying i) and ii). We proceed by casework on $1-\hat{f}[0]$.

\begin{case}
$1-\hat{f}[0]\ge\varepsilon$.
\end{case}

Define $f_{\ell}$ (resp. $f'_u$) to be $f$ but with each child $c$ of $f$ replaced by $c_{u}$ (resp. $c_{\ell}$). Then $$(1-\hat{f_{\ell}}[0])-(1-\hat{f}[0]) = \prod_{c_{u}} \hat{c_u}[0]-\prod_{c} \hat{c}[0]\le O\left(\cdot\sum_{c_{u}}L(c_{u})\sqrt{\varepsilon}\right)\le O(L(f)\sqrt{\varepsilon})$$ The same analysis tells us $(1-\hat{f}[0]) - (1-\hat{f'_u}[0])\le O(L(f)\sqrt{\varepsilon})$. If $f'_u \ge \varepsilon$, take $f_u$ to be $f'_u$; otherwise, take $f_u$ to be the constant 1 function, in which case $$(1-\hat{f}[0])-(1-\hat{f_u}[0])\le O(L(f)\sqrt{\varepsilon}) + \varepsilon\le O(L(f)\sqrt{\varepsilon}).$$ It follows that $f_{\ell}$ and $f_u$ are $\sqrt{\varepsilon}$-sandwiching formulas for $f$ which satisfy ii) by construction.

It remains to verify i). Assume $f_{\ell}$ and $f_u$ are nonconstant. For $f_{\ell}$, we know $1-\hat{f_{\ell}}[0]\ge 1-\hat{f}[0]\ge\varepsilon$, and $1 - \hat{f_{\ell}}[0]\le 1-\varepsilon$ because $\hat{c_u}[0] = \le 1 - \varepsilon$ for all nonconstant children $c_u$ of $f_{\ell}$. For $f_u$, by construction, $(1-\hat{f_u}[0])\ge\varepsilon$, and $1-\hat{f_u}[0]\le 1-\hat{f_{\ell}}[0]\le 1-\varepsilon$.

\begin{case}
$1-\hat{g}[0]<\varepsilon$.
\end{case}

Define $f_u$ to be the constant 1 function. Define $f'_{\ell}$ to be $f$ but with each child $c$ of $f$ replaced by $c_u$. If $f'_{\ell} \ge\varepsilon$, take $f_{\ell}$ to be $f'_{\ell}$.

Otherwise, we note that it's possible to prune from $f'_{\ell}$ enough children to get $f_{\ell}$ such that $\varepsilon\le 1-\hat{f_{\ell}}[0]\le\sqrt{\varepsilon}$. Assume to the contrary. Order the children $c_u$ in any way $\{c_1,...,c_k\}$ and define $q^j = \prod^k_{i=j}1-(1-\hat{c_i}[0])$. Then $q^1 > \sqrt{\varepsilon}$ and $q^k <\varepsilon$. Then either there is some $j$ for which $\varepsilon\le q^j\le\sqrt{\varepsilon}$, or there is some $j$ for which $\varepsilon\le 1-\hat{c_j}[0]\le\sqrt{\varepsilon}$, a contradiction.

By construction, $f_{\ell}$ and $f_u$ are sandwiching formulas for $f$ which satisfy ii).

It remains to verify i). $f_u$ is constant. For $f_{\ell}$, $1-\hat{f_{\ell}}[0]\ge\varepsilon$ by construction. If $f_{\ell} = f'_{\ell}$, then because $1-\hat{f_{\ell}}[0]\le 1-\varepsilon$ for the same reason as in Case 1. Otherwise, we know by construction that $1-\hat{f_{\ell}}[0]\le\sqrt{\varepsilon}< 1-\varepsilon$.
\end{proof}


\begin{thebibliography}{99}
\bibitem{aw} Miklos Ajtai and Avi Wigderson. Deterministic simulation of probabilistic constant depth circuits. Advances in Computing Research - Randomness and Computation, 5:199-223, 1989. Preliminary version in Proc. of FOCS’85. 1
\bibitem{aghp} N. Alon, O. Goldreich, J. H{\aa}stad, and R. Peralta. Simple constructions of almost k-wise independent random variables. Random Structures \& Algorithms, 3(3):289-304, 1992. See also addendum in issue 4(1), 1993, pp. 199-120.
\bibitem{amn} Y. Azar, R. Motwani, and J. Naor. Approximating probability distributions using small sample spaces. Combinatorica, 18(2):151-171, 1998.
\bibitem{baz} Louay Bazzi. Polylogarithmic independence can fool DNF formulas. In Proceedings of the 48th IEEE Symposium on Foundations of Computer Science, pages 63-73, 2007.
\bibitem{bdvy} A. Bogdanov, Z. Dvir, E. Verbin, and A. Yehudayoff. Pseudorandomness for Width 2 Branching Programs.
Electronic Colloquium on Computational Complexity (ECCC), 16:70, 2009.
\bibitem{bpw} A. Bogdanov, P. A. Papakonstantinou, and A. Wan. Pseudorandomness for read-once formulas. In FOCS, 2011, pp. 240-246.
\bibitem{brav} Mark Braverman. Poly-logarithmic independence fools AC0 circuits. Technical Report TR09-011, Electronic Colloquium on Computational Complexity, 2009.
\bibitem{bv} J. Brody and E. Verbin. The coin problem, and pseudorandomness for branching programs. In \textit{Proceedings of the fifty first annual symposium on Foundations of Computer Science (FOCS)}, 2010.
\bibitem{dett} Anindya De, Omid Etesami, Luca Trevisan, and Madhur Tulsiani. Improved pseudo- random generators for depth 2 circuits. In APPROX-RANDOM, pages 504-517, 2010.
\bibitem{gmrtv} Parikshit Gopalan, Raghu Meka, Omer Reingold, Luca Trevisan, and Salil Vadhan. Better pseudorandom generators from milder pseudorandom restrictions. In FOCS, pages 120-129, 2012.
\bibitem{hastshrink} J. H{\aa}stad, The shrinkage exponent of de Morgan formulas is 2, SIAM J. Comput. 27 (1998), no. 1, 48-64.
\bibitem{hastparity} J. H{\aa}stad, ``Computational limitations for small depth circuits'', Ph.D. thesis, M.I.T. press, 1986.
\bibitem{hry} Johan H{\aa}stad, Alexander A. Razborov, and Andrew Chi-Chih Yao, On the shrinkage exponent for read-once formulae, Theor. Comput. Sci. 141 (1995), no. 1\&2, 269-282.
\bibitem{ik} R. Impagliazzo and V. Kabanets. Fourier concentration from shrinkage. Electronic Colloquium on Computational Complexity (ECCC), 20:163, 2013. To appear in CCC 2014.
\bibitem{imz} Russell Impagliazzo, Raghu Meka, and David Zuckerman. Pseudorandomness from shrinkage. In Proceedings of the 53rd IEEE Symposium on Foundations of Computer Science, 2012.
\bibitem{iw} R. Impagliazzo and A. Wigderson. $\P=\BPP$ if $\E$ requires exponential circuits: Derandomizing the XOR Lemma. In Proceedings of the Twenty-Ninth Annual ACM Symposium on Theory of Computing, pages 220-229, 1997.
\bibitem{lmn} N. Linial, Y. Mansour, and N. Nisan. Constant depth circuits, Fourier transform and learnability. J. ACM, 40(3):607-620, 1993.
\bibitem{mansour} Yishay Mansour. An o(nlog log n) learning algorithm for DNF under the uniform distribution. \emph{Journal of Computer and System Sciences}, 50(3):543–550, 1995. 3, 8, 10
\bibitem{nis} N. Nisan. Pseudorandom bits for constant depth circuits. Combinatorica, 12(4):63-70, 1991.
\bibitem{nn} Joseph Naor and Moni Naor, Small-bias probability spaces: Efficient constructions and applications, SIAM J. Comput. 22 (1993), no. 4, 838-856.
\bibitem{ODonnell} R. O’Donnell. Some topics in analysis of boolean functions. \emph{Proc. STOC} 2008, 569–578, 2008.
\bibitem{raz} Alexander Razborov. A simple proof of Bazzi’s theorem. ACM Trans. Comput. Theory, 1(1):1-5, 2009.
\bibitem{rsv} Omer Reingold, Thomas Steinke, and Salil Vadhan. Pseudorandomness for regular branching programs via fourier analysis. In APPROX-RANDOM, pages 655-670, 2013.
\bibitem{rod} V. R{\"o}dl, On a packing and covering problem, Europ. J. Combinatorics 6 (1985), 69–78.
\bibitem{sss} Jeanette P. Schmidt, Alan Siegel, and Aravind Srinivasan, Chernoff-Hoeffding bounds for applications with limited independence, SIAM J. Discrete Math. 8 (1995), no. 2, 223-250.
\bibitem{svw} Thomas Steinke, Salil P. Vadhan, and Andrew Wan, Pseudorandomness and fourier growth bounds for width 3 branching programs, CoRR abs/1405.7028 (2014).
\bibitem{tal} A. Tal. Shrinkage of de Morgan formulas from quantum query complexity. Electronic Colloquium on Computational Complexity, 21(48), 2014.
\bibitem{tx} L. Trevisan and T. Xue. A derandomized switching lemma and an improved derandomization of AC0. In Proceedings of the Twenty-Eighth Annual IEEE Conference on Computational Complexity, pages 242-247, 2013.
\end{thebibliography}
\end{document}